\documentclass[AMA,STIX1COL]{WileyNJD-v2}
\articletype{Research Article}%

\received{26 April 2016}
\revised{6 June 2016}
\accepted{6 June 2016}

\usepackage{lscape,rotating}
\usepackage{float,epsfig}
\usepackage{subfigure}
\usepackage{graphicx}

\newtheorem{example}{Example}

\newcommand{\thetav}{\mathbf{\theta}}

\begin{document}

\title{MOVER confidence intervals for a difference or ratio effect parameter under stratified sampling}

\author{Yongqiang Tang}

\address{Department of Biometrics, Grifols, Durham, NC, USA}

\corres{Yongqiang Tang, Boston, USA \newline
E-mail: yongqiang\_tang@yahoo.com}

\abstract[Summary]{
Stratification is commonly employed in clinical trials to reduce the chance covariate imbalances and increase  the precision of the treatment effect estimate.
We propose a general framework for constructing the confidence interval (CI) for  a difference or ratio effect parameter under stratified sampling 
by  the method of variance estimates recovery (MOVER).  
We consider the additive variance and additive CI approaches for the difference, in which 
 either the CI for the weighted difference, or the CI for the weighted effect in each group, or the variance for the weighted difference is 
calculated as the weighted sum of the corresponding stratum-specific statistics.
The CI for the ratio is derived by the Fieller and log-ratio methods. 
The weights can be random quantities under the assumption of a constant effect across strata, but this assumption is not needed for fixed weights.
These methods can be easily applied  to different endpoints in that they
require only  the point estimate, CI, and  variance estimate for the measure of interest  in each group across strata.
 The methods are illustrated with two real  examples.
In one example, we derive the MOVER CIs for the risk difference and risk ratio  for binary outcomes. In the other example, we  compare the restricted mean survival time and milestone survival in stratified analysis of time-to-event outcomes.
 Simulations show that the proposed MOVER CIs generally outperform the standard large sample CIs, and that the additive CI approach performs better than the additive variance approach.  Sample SAS code is provided in the Supplementary Material.
}
\keywords{
Additive confidence interval approach; Additive variance approach; Delta method; Fieller method;  Mantel-Haensze estimator; Minimum risk weight; Non-constant effect; Restricted mean survival time}

\maketitle

\section{Introduction}
In clinical trials, randomization is often performed by stratifying on a few  prognostic factors.
 For example, in the belimumab trial for systemic lupus erythematosus \cite{stohl:2017}, subjects were stratified by  a screening SELENA–SLEDAI score ($\leq$ 9 versus $\geq$10), complement level (those with versus those without low C3 and/or C4), and race (black versus non‐black).  
Literature reviews have been conducted to assess the usefulness of stratified randomization \cite{kernan:1999}, and the reporting and analyses of stratified clinical trials \cite{kernan:2012}.
Stratification serves two purposes  \citep{chmp:2015}.  First, it prevents imbalances  between treatment groups in  stratification factors.
When there is a chance imbalance in important prognostic factors under simple randomization,   the response could be different between treatment groups even if the two treatments have the same effect \cite{chu:2012}.
The treatment effect estimate without adjustment for such imbalance tends to be biased. 
  Second, stratification improves the precision of the treatment effect estimate \cite{grizzle:1982,mchugh:1983, miratrix:2013, tang:2018, tang:2020b}.  Although the post-stratification
is as efficient as the pre-stratification in large samples  \citep{peto:1977}, pre-stratification can be much more efficient than the post-stratification in small and moderate samples \citep{grizzle:1982,mchugh:1983, miratrix:2013, tang:2018, tang:2020b},
and permit meaningful subgroup analyses \citep{chmp:2015} since randomization is applied to each stratum.

The   method of variance estimates recovery (MOVER), originally developed by  Howe \cite{howe:1974}  for constructing the confidence interval (CI) for the mean of the sum of two independent variables, 
 becomes popular since it was generalized to the difference and ratio effect parameters  \citep{newcombe:1998,donner:2012, newcombe:2016}. 
Sometimes, the Wald approach might be the only simple direct  interval estimation method  in complex situations such as
the comparison of the restricted mean survival time (RMST) on the basis of the nonparametric Kaplan-Meier (KM) technique. However, it is straightforward to use the MOVER technique,
in which the CI for RMST is derived for each group, and combined  into the CI for the difference or ratio of RMST between two groups \citep{tang:2021}.
One  selects a single sample CI  (e.g.  score type CIs) with good properties, which will be inherited by the MOVER CI for the comparison of two groups. 
The MOVER method generally performs well compared to other asymptotic approaches in finite samples \citep{newcombe:1998, lee:2004, donner:2012, tang:2021}.

The purpose of the paper  is to propose a general framework for constructing MOVER CIs for a  difference and ratio effect parameter under stratified sampling.
We derive the  additive variance (AV) and additive CI (AC)  approaches for the difference parameter.
They require only  the point estimate, CI, and  variance estimate for the parameter of interest  in each group across strata.
The  AV approach relies on the delta method, in which the variance for the weighted difference is calculated as the sum of the MOVER variance for the difference in each stratum. Although 
the variability in the weight is ignored, the AV approach is asymptotically valid for random weights such as the inverse variance (INV) weight when the effect is constant across strata.

In the AC approach, the CI of the weighted sum statistic is calculated as the sum of the  stratum-specific $(1-\gamma)$-level CI for an appropriate $\gamma$ by ignoring the variability in the weights. 
There are two variations (labeled as ``AC'' and ``AC2''). 
In AC, the CI is constructed for the summary effect in each group, and then combined into the MOVER CI for the difference between two groups.
In AC2, we derive the MOVER or other CIs for the stratum-specific difference, and combine them into the CI for the weighted difference.

The CI for the ratio  is obtained by the Fieller or  log-ratio method. A concern with the Fieller CI is that it may be disjoint, and hence difficult to interpret \citep{sherman:2011}.
We show the Fieller CI  is non-disjoint and uniquely determined if the parameter of interest
 takes only positive values, and this applies to  binary proportions, Poisson rates or mean survival.
The proposed MOVER CIs are asymptotically valid for random weights  if the difference or ratio is constant across strata, or for non-constant effects when the weights are fixed.

The rest of the paper is organized as follows. 
We introduce the general MOVER methods under stratified sampling in Section 2.
 Real data analyses and simulations are given respectively in sections 3 and 4.

\section{MOVER intervals  under stratified sampling}

\subsection{Review of MOVER CIs in unstratified samples}
We are interested in the risk difference (RD)   between two independent groups ($g=0,1$).
The $100(1-\alpha)\%$ CI for $\tau_g$ is $(l_g,u_g)$. 
The MOVER Cl for $d=\tau_1-\tau_0$ is constructed as
\begin{equation}\label{ciunstra}
\left[\hat{d}-\sqrt{(l_1 -\hat{\tau}_1)^2+(u_0 -\hat{\tau}_0)^2}, \hat{d}+\sqrt{(u_1 -\hat{\tau}_1)^2+(l_0 -\hat{\tau}_0)^2}\right],
\end{equation}
where the symbol with $\,\hat{}\,$ denotes the point estimate. Let   $\sigma_g^2$ be the variance of $\hat\tau_g$.
The underlying idea  \citep{donner:2012} is that the variance $\sigma_g^2$ can be recovered as $( l_g - \hat{\tau}_{g})^2/ z_{\alpha/2}^2$ or
 $(u_g - \hat{\tau}_{g})^2/ z_{\alpha/2}^2$, where $z_p$ is the $(1-p)$-th percentile of the normal distribution.
One shall choose a one-sample CI with good properties. Examples include the score type CI for binary \citep{newcombe:1998}, Poisson \citep{li:2014} and survival \citep{tang:2021} outcomes, and 
the exact CI for variance components related to the multivariate
normal distribution \citep{lee:2004}.
The MOVER CI in Equation \eqref{ciunstra} is asymptotically equivalent to the Wald CI, but the MOVER CI generally has a better performance in finite samples.

There are possible different ways to construct the MOVER CI. Two simple alternatives are
$\left[\hat{d}-\sqrt{(u_1 -\hat{\tau}_1)^2+(l_0 -\hat{\tau}_0)^2}, \hat{d}+\sqrt{(l_1 -\hat{\tau}_1)^2+(u_0 -\hat{\tau}_0)^2}\right]$, and 
$[\hat{d}-\sqrt{\bar{\sigma}^2}, \hat{d}+\sqrt{\bar{\sigma}^2}]$, where $\bar{\sigma}^2=\sum_{g=0}^1[(l_g -\hat{\tau}_g)^2+(u_g -\hat{\tau}_g)^2]/2$.
We will use formula \eqref{ciunstra} since it reduces to the one-sample CI $[l_1,u_1]$  when  $\hat\tau_0=l_0=u_0=0$, and to $[l_1-\tau_0,u_1-\tau_0]$
as $n_0\rightarrow \infty$. 
In unreported simulations, formula \eqref{ciunstra} outperforms the two alternatives for binary outcomes.

\subsection{Weighted difference}\label{dmover}
As shown in Appendix \ref{biaseff}, the unstratified RD estimate  is generally biased under stratified sampling. 
We describe a general framework for constructing the CI for a difference or ratio effect parameter in stratified studies. Subjects in different strata and groups are assumed to be independent. 
 Let $\thetav_s$ denote the set of model parameters for stratum $s=1,\ldots, S$, and $\thetav=\{\thetav_1,\ldots,\thetav_S\}$.
Suppose we want to construct the CI  for the weighted difference
\begin{equation}\label{dmeasure}
\tau=\sum_{s=1}^S w_s [\delta_{s1}-\phi \delta_{s0}]=\tau_1-\phi\tau_0,
\end{equation}
 where $\delta_{sg}$ is a function of $\thetav_s$, and $\tau_g=\sum_s w_s\delta_{sg}$. We set $\phi=1$  if the difference is of interest. 
In general, $\phi$ denotes the ratio $\tau_1/\tau_0$.
In the Fieller CI approach for the ratio, the CI for $\tau$ is needed  given  $\phi$, and  the details will be given in Section \ref{rmover}.
We assume $\sum_{s=1}^S w_s=1$ although the constraint is not necessary in calculating the ratio  $\phi$.

Let $n_{sg}$ be the sample size in group $g$ stratum $s$, 
 $\hat\sigma_{sg}^2$ the variance for $\hat{\delta}_{sh}$, $(l_{sg},u_{sg})$  the CI for $\delta_{ig}$,  
and $\hat{w}_s=w_s(\hat\thetav)$  the estimated weight. The total sample size is $n=\sum_{g=0}^1\sum_{s=1}^S n_{sg}$.

We illustrate the problem by a bioassay study \cite{gart:1985} in evaluating the carcinogenic effect on four sex-strain groups of mice. 
The number of responders $x_{sg}$ and  mice  $n_{sg}$  in the four strata are presented below 
\begin{center}
\begin{tabular}{l@{\extracolsep{4pt}}c@{\extracolsep{4pt}}c@{\extracolsep{4pt}}c@{\extracolsep{4pt}}c@{\extracolsep{3pt}}c}\hline
 stratum & 1 & 2 & 3 & 4\\\hline
control $x_{s0}/n_{s0}$  & {5}/{79} & {3}/{87} & {10}/{90} & {3}/{82}  \\
treated $x_{s1}/n_{s1}$ & {4}/{16} & {2}/{16} & {4}/{18} & {1}/{15} \\
\hline
\end{tabular}
\end{center}
\vspace{0.05in}

Let $p_{sg}$ be the proportion in group $g$ strata $s$. 
Then  $\thetav_s=\{p_{s0},p_{s1}\}$ and $\delta_{sg}=p_{sg}$. 
The weighted RD estimator is given by 
\begin{equation}\label{rdestbin}
 \tau =\sum_{s=1}^S w_s [\delta_{s1}- \delta_{s0}]= \sum_{s=1}^S w_s [p_{s1}- p_{s0}]
\end{equation}
with the MH, INV and MR weights. The MH weight $w_s^{(mh)}\propto n_{s1}n_{s0}/n_s$ is fixed while the INV and MR weights are random quantities depending on $\hat{p}_{sg}$'s.
The MOVER CIs are constructed on the basis of the single sample estimates
 $\hat{p}_{sg}=x_{sg}/n_{sg}$, $\hat\sigma_{sg}^2=\hat{p}_{sg}(1-\hat{p}_{sg})/n_{sg}$ and the Wilson score interval\citep{newcombe:1998} for $p_{sg}$.

The AV, AC and AC2  CIs for $\tau$ rely on two lemmas and one corollary established in the remaining of this section.  Their proofs are given in the appendix.
The weight  $\hat{w}_s=w_s(\hat\thetav)$ shall not depend on
$n$ after normalization by the total weight. Otherwise, the Taylor series approximation underlying the delta method for Lemma   \ref{varsigle} may fail.
For this reason, Lemma   \ref{varsigle} is not suitable for the MR weight, and this will be illustrated in  Section \ref{secsim}.

\begin{lemma}\label{varsigle}
Assume both functions $\hat{w}_s=w_s(\hat\thetav)$ and  $\hat\delta_{sg} =\delta_{sg}(\hat\thetav_s)$ are  continuous and differentiable at  $\thetav$. 
Suppose $\sqrt{n}(\hat{\thetav}-\thetav)$ has a limiting  normal distribution.\\ 
a.) The asymptotic  variance of $\hat\tau_g^* =\sum_{s=1}^S\hat{ w}_s \,[\hat{\delta}_{sg}- \delta_{sg}] $ is given by
$\mbox{var}(\hat\tau_g^*)=\sum_{s=1}^S w_s^2 \,\mbox{var}(\hat{\delta}_{sg}).$
It is also the variance for $\hat\tau_g=\sum_{s=1}^S \hat{w}_s \hat{\delta}_{sg}$ when $\hat{w}_s$'s are fixed, or when $\delta_{1g}=\ldots=\delta_{Sg}$.\\
b.) 
If $\hat{w}_s$'s are fixed, or if $\delta_{s1}=\phi \delta_{s0}$ in all strata for a known $\phi$,
the asymptotic  variance of $\hat\tau = \sum_{s=1}^S \hat{w}_s \,[\hat{\delta}_{s1}- \phi  \hat{\delta}_{s0} ]=\hat\tau_1-\hat\tau_0$ is
$$\text{var}(\hat\tau)=\sum_{s=1}^S w_s^2 \left[\text{var}(\hat{\delta}_{s1})+\phi^2\text{var}(\hat{\delta}_{s0})\right] =\text{var}(\hat\tau_1^*)+\phi^2 \text{var}(\hat\tau_0^*).$$
c.) Under the assumption in (b),   the asymptotic variance of  $\hat{\psi}=\log[ \sum_{s=1}^S \hat{w}_s \,\hat{\delta}_{s1}] - \log[\sum_{s=1}^S \hat{w}_s \,\hat{\delta}_{s0}]$ is given by
$$\text{var}(\hat{\psi}) =  \frac{\text{var}(\hat\tau_1^*)}{\tau_1^{2}} + \frac{ \text{var}(\hat\tau_0^*)}{\tau_0^{2}}. $$
\end{lemma}

While Lemma \ref{varsigle}b seems obvious when $w_i$'s are fixed, its merit lies in its validity for a common difference or ratio parameter across strata when the weights 
 are random quantities. Lemma \ref{varsigle}c gives the variance for the log ratio of two weighted means. 

In the AV approach,  the variance of $\hat\tau = \sum_{s=1}^S \hat{w}_s \,[\hat{\delta}_{s1}- \phi \hat{\delta}_{s0} ]$ is recovered
 in an additive form as $\hat\sigma_L^2=\sum_{s=1}^S \hat{w}_s^2 [(l_{s1}- \hat{\delta}_{s1})^2 +\phi^2 (u_{s0}- \hat{\delta}_{s0})^2]$ or 
$\hat\sigma_U^2=\sum_{s=1}^S \hat{w}_s^2 [(u_{s1}- \hat{\delta}_{s1})^2 +\phi^2 (l_{s0}- \hat{\delta}_{s0})^2]$ according to Lemma \ref{varsigle}b. 
The AV CI for $\tau$ is
$$\left[\hat\tau -\sqrt{\hat\sigma_L^2}, \hat\tau +\sqrt{\hat\sigma_U^2}\right].$$

\begin{lemma}\label{eachci}
Suppose the $1-\gamma$ level CI for $\delta_{sg}$ is $(l_{sg}^{(\gamma)}, u_{sg}^{(\gamma)})$, where  $\sigma_{sg}^2$ is the  variance of $\hat{\delta}_{sg}$, 
$\hat\sigma_{sg}^2$ is a consistent  variance estimate, and
$z_{\gamma/2} = \frac{\sqrt{\sum_s \hat{w}_s^2 \hat\sigma_{sg}^2}}{\sum_s \hat{w}_s \hat\sigma_{sg}} z_{\alpha/2}$. 
Let $(L_g,U_g)=(\sum_s \hat{w}_s l_{sg}^{(\gamma)}, \sum_s \hat{w}_s u_{sg}^{(\gamma)})$. 
If $\sigma_{sg}^2$ is consistently estimated by $(l_{sg}-\hat{\delta}_{sg})^2/z_{\gamma/2}^2$ and $(u_{sg}-\hat{\delta}_{sg})^2/z_{\gamma/2}^2$, then\\
a)  $(L_g-\hat\tau_g)^2 =  z_{\alpha/2}^2\sum_s w_s^2 \sigma_{sg}^2+o_p(n^{-1})$ and  $ (U_g  -\hat\tau_g)^2=  z_{\alpha/2}^2\sum_s w_s^2 \sigma_{sg}^2+o_p(n^{-1})$\\
b) Both $[\log(L_g)- \log(\hat\tau_g)]^2$ and $[\log(U_g)- \log(\hat\tau_g) ]^2 $ are consistent estimators of $ z_{\alpha/2}^2 \sum_s {w}_s^2 \sigma_{sg}^2/\tau_g^{2}$.
\end{lemma}
By Lemma \ref{eachci}a and Lemma \ref{varsigle}b,  the $100(1-\alpha)\%$ MOVER CI for $\tau=\tau_1-\phi\tau_0$ is given by
\begin{equation}\label{acformula}
\left[\hat{\tau}-\sqrt{(L_1-\hat\tau_1)^2+\phi^2(U_0-\hat\tau_0)^2},\hat{\tau}+\sqrt{(U_1-\hat\tau_1)^2+\phi^2(L_0-\hat\tau_0)^2}\right].
\end{equation}
This CI is referred to as the AC CI since 
  $[L_g,U_g]$  is the $100(1-\alpha)\%$ CI for $\tau_g$, and  the weighted sum of  the stratum-specific $100(1-\gamma)\%$ CIs for  $\delta_{sg}$  when the weights are fixed or 
 there is no stratum effect ($\delta_{1g}=\ldots=\delta_{Sg}$).
However,  when there is a stratum effect, the CI in Equation \eqref{acformula} is still valid for random weights if the difference $\delta_{s1}-\phi\delta_{s0}$ is constant. 

Lemma \ref{eachci}a is a generalization of Yan and Su (YS \cite{yan:2010}) result for binomial proportions.  There are several issues in the YS method \cite{yan:2010}. First, 
 $(L_g,U_g)$ is treated as the $100(1-\alpha)\%$ CI for $\tau_g$, and this is incorrect for random weights $w_i$'s when there is stratum effect.  
The YS MOVER CI takes the form
$$ \left[\hat\tau -z_{\alpha/2} \sqrt{\hat\sigma_{yL}^2}, \hat\tau -z_{\alpha/2} \sqrt{\hat\sigma_{yU}^2}\right],$$
where $\lambda_g=\sum_{s=1}^S w_s^2/n_{gs}$, $\hat\sigma_{yL}^2 = \lambda_1L_1(1-L_1)+\lambda_0U_0(1-U_0)$, and $\hat\sigma_{yU}^2 = \lambda_1U_1(1-U_1)+\lambda_0L_0(1-L_0)$.
The YS CI tends to overcover (conservative)  if the response rates vary greatly  across strata, and this will be shown  in Section \ref{secsim}.
YS \cite{yan:2010} suggested optimal weights for each group. When the weights differ between two groups, the CI defined in Equation \eqref{acformula} may be invalid  for a constant difference when there is a stratum effect.

\begin{corollary}\label{corac2}
Let $\hat\sigma_s^2$ be a consistent variance estimate for $\hat\delta_{s1}-\phi\hat\delta_{s0}$,
and $(\mathcal{L}_{s}, \mathcal{U}_{s})$ the $1-\gamma$ level CI for $\delta_{s1}-\phi\delta_{s0}$, where
$z_{\gamma/2} = \frac{\sqrt{\sum_s \hat{w}_s^2 \hat\sigma_s^2 }}{\sum_s \hat{w}_s\hat\sigma_{s} } z_{\alpha/2}$. 
Then $[\mathcal{L},\mathcal{U}]=[\sum_s \hat{w}_s \mathcal{L}_{s}, \sum_s \hat{w}_s \mathcal{U}_{s}]$ is the $100(1-\alpha)\%$ CI of $\tau$ if the weights $\hat{w}_s$'s
are fixed, or if the difference $\delta_{s1}-\phi\delta_{s0}$ is constant across strata.
\end{corollary}
The AC2 CI is obtained from Corollary \ref{corac2}. It combines the MOVER or other CIs for the  difference in each stratum into a single CI for $\tau$. 
In the application of Lemma \ref{eachci} and Corollary \ref{corac2}, 
we employ the delta variance for $\hat\delta_{sg}$ and $\hat\delta_{s1}-\phi\hat\delta_{s0}$. Alternatively, one may use the variance  [e.g. $\frac{( l_{sg} - \hat{\tau}_{sg})^2}{ z_{\alpha/2}^2}$, 
 $\frac{(u_{sg} - \hat{\tau}_{sg})^2}{z_{\alpha/2}^2}$, $\frac{(u_{sg}-l_{sg})^2}{4z_{\alpha/2}^2}$] recovered from the CI.

\subsection{MOVER CIs for ratio under stratified sampling}\label{rmover}

The CI for the ratio is obtained via the Fieller or log-ratio method. In the log-ratio method, we first derive the CI for $\log(\tau_1)-\log(\tau_0)$ and then exponentiate it.
The Fieller approach inverts the test $H_0: \tau_1 -\phi \tau_0=0$.  The Fieller CI contains all values of $\phi$ for which the null hypothesis  is not rejected, 
or equivalently the CI for $\tau_1-\phi\tau_0$ includes $0$. 

\subsubsection{AC approach}\label{rdmover}
The Fieller CI  contains all  values of $\phi$ for which the CI for $\tau_1- \phi \tau_0$ given in Equation \eqref{acformula} includes 0.
If the parameters $\tau_{g}$'s take only positive values, the Fieller  CI is uniquely determined and non-disjoint. It is a consequence of the following lemma
 \begin{lemma}\label{fieller}
 The function $Z(\phi) = \frac{\tau_1 -\phi \tau_0}{ \sqrt{V_1+ \phi^2 V_0}}$ ($\tau_1\geq 0$, $\tau_0\geq 0$, $V_0>0$ and $V_1>0$)
 decreases from  $\frac{\tau_1}{ \sqrt{V_1}}$ to $- \frac{\tau_0}{ \sqrt{V_0}}$
 as $\phi$ increases from $0$ to $\infty$ because $\frac{\partial Z(\phi)}{\partial \phi}<0$. There is at most one solution for $z(\phi)=z_c$, where $z_c=\pm 1$ or $\pm z_{\alpha/2}$.
 \end{lemma}
  
In MOVER, setting the lower and upper limits for  $\tau_1- \phi \tau_0$ to 0 yields the CI for the ratio,
\begin{equation}\label{fiellerac}
\phi_l =\frac{b-\sqrt{b^2-a_lc_l}}{a_l}, \phi_u = \frac{b+\sqrt{b^2-a_uc_u}}{a_u}
\end{equation}
where $b=\hat{\tau}_1\hat{\tau}_0$, $a_l =\hat{\tau}_0^2- (U_0-\hat{\tau}_0)^2$, $c_l=\hat{\tau}_1^2-(L_1-\hat{\tau}_1)^2$,
 $a_u = \hat{\tau}_0^2- (L_0-\hat{\tau}_0)^2$, and $c_u=\hat{\tau}_1^2-(U_1- \hat{\tau}_1)^2$. If $a_u=0$, then $\phi_u=\infty$. 
We have $c_l\geq 0$ and $a_u\geq 0$. The solution \eqref{fiellerac} always exists since $b^2- a_l c_l\geq 0$ and $b^2-a_uc_u\geq 0$ no matter whether $a_l\geq 0$ and $c_u\geq0$.

In the log-ratio method,  the $100(1-\alpha)\%$ MOVER CI for $\log(\tau_1)-\log(\tau_0)$ is obtained based on  Lemma \ref{eachci}b and Lemma \ref{varsigle}c, and then exponentiated
\begin{equation}\label{aclogformula}
 \left[\exp\left(\log(\frac{\hat\tau_1}{\hat\tau_0})-\sqrt{\hat\sigma_{lL}^2}\right),\exp\left(\log(\frac{\hat\tau_1}{\hat\tau_0})+\sqrt{\hat\sigma_{lU}^2 }\right)\right],
\end{equation}
where $\hat\sigma_{lL}^2=\log^2(L_1/\hat\tau_1)+\log^2(U_0/\hat\tau_0)$, and $\hat\sigma_{lU}^2=\log^2(U_1/\hat\tau_1)+\log^2(L_0/\hat\tau_0)$. The log-ratio CI is incomputable if $\hat\tau_0=0$ and $\hat\tau_1=0$.
Otherwise it is usually  similar to the Fieller CI.

We refer the Fieller and log-ratio CIs to respectively as ``AC'' and ``ACL''.
It is valid for random $w_i$'s with a constant ratio. Interestingly, the ratio does not have to be the same across strata for fixed (e.g. MH) weights.

\subsubsection{AC2 approach}\label{rdmover2}
 Setting the limits of the AC2 CI for $\tau_1-\phi\tau_0$ to 0 yields the Fieller CI for $\phi$.  The lower limit is $\phi_l=0$ if $\hat\tau_1=0$.
The upper limit is $\phi_u=\infty$ if $\hat\tau_0=0$. In general, there is no explicit analytic solution for the confidence limits.
They may be solved by the bisection method by using the AC confidence limit as the initial value.

\subsubsection{AV approach}\label{rdmover}
In the Fieller method, the CI has a similar form to Equation \eqref{fiellerac},
 \begin{equation}\label{fiellerav}
\phi_l = \frac{b-\sqrt{b^2-a_lc_l}}{a_l}, \phi_u = \frac{b+\sqrt{b^2-a_uc_u}}{a_u}
\end{equation}
 where $b=\hat{\tau}_1\hat{\tau}_0$,
 $a_l = \hat{\tau}_0^2-\sum_s \hat{w}_s^2 (u_{s0}-\hat{\delta}_{s0})^2 $,  $c_l=\hat{\theta}_1^2-\sum_s \hat{w}_s^2 (l_{s1}-\hat{\delta}_{s1})^2$,
$a_u =\hat{\tau}_0^2-\sum_s \hat{w}_s^2 (l_{s0}-\hat{\delta}_{s0})^2 $, and $c_u=\hat{\tau}_1^2-\sum_s \hat{w}_s^2 (u_{s1}-\hat{\delta}_{s1})^2$.

In the log-ratio method, the variance of $\log(\hat{\tau}_1/\hat\tau_0)$ is given by Lemma \ref{varsigle}c, and the CI is 
\begin{equation}\label{ratiolarge}
\left[\exp\left(\log(\frac{\hat{\tau}_1}{\hat\tau_0})-z_{\alpha/2}\sqrt{\frac{\sigma_1^2}{\hat\tau_1^2} + \frac{\sigma_0^2}{\hat\tau_0^2}}\right),  \exp\left(\log(\frac{\hat{\tau}_1}{\hat\tau_0})+z_{\alpha/2}\sqrt{\frac{\sigma_1^2}{\hat\tau_1^2} + \frac{\sigma_0^2}{\hat\tau_0^2}}\right)\,\right],
\end{equation}
where $\sigma_g^2$ is the asymptotic or MOVER variance for $\hat\tau_g^*$. In the MOVER approach, we replace $z_{\alpha/2}^2\sigma_g^2$ by
$\sum_s \hat{w}_s^2 (l_{sg}-\hat\theta_{sg})^2$ or $\sum_s \hat{w}_s^2 (u_{sg}-\hat\theta_{sg})^2$.

We denote the  Fieller and log-ratio CIs by ``AV'' and ``AVL'' respectively.

\section{Data examples}

\begin{table}[h]
\centering{
\begin{tabular}{c@{\extracolsep{5pt}}c@{\extracolsep{5pt}}c@{\extracolsep{5pt}}c@{\extracolsep{5pt}}c@{\extracolsep{5pt}}c@{\extracolsep{5pt}}c@{\extracolsep{5pt}}c@{\extracolsep{5pt}}c@{\extracolsep{5pt}}c@{\extracolsep{5pt}}c@{\extracolsep{5pt}}c@{\extracolsep{5pt}}}\\\hline
\multicolumn{4}{c}{Risk Difference} &     \multicolumn{3}{c}{Relative risk}  \\  \cline{1-4} \cline{5-6} 
 & MH$^{\S}$& INV & MR &   &   MH$^{\S}$ \\ \hline

 Est.&$ 0.106$&$ 0.084$&$ 0.096$& &$ 2.674$\\
\multicolumn{6}{l}{\bf Asymptotic CIs}\\
                                             DC$^{\dagger}$&$[ 0.012, 0.200]$&$-$&$-$&DC$^{\dagger}$&$[ 1.366, 5.234]$\\
                                             Wald&$[ 0.013, 0.198]$&$[-0.001, 0.169]$&$[ 0.005, 0.187]$&ASY&$[ 1.369, 5.222]$\\

\multicolumn{6}{l}{\bf MOVER CIs}\\
                                             AV&$[ 0.038, 0.225]$&$[ 0.025, 0.211]$&$[ 0.034, 0.217]$&AV&$[ 1.442, 5.033]$\\
                                             YS&$[ 0.027, 0.217]$&$[ 0.006, 0.200]$&$[ 0.015, 0.211]$&AVL&$[ 1.370, 5.688]$\\
                                             AC&$[ 0.029, 0.216]$&$[ 0.016, 0.190]$&$[ 0.022, 0.206]$&AC&$[ 1.373, 5.093]$\\
                                             AC2&$[ 0.029, 0.216]$&$[ 0.016, 0.190]$&$[ 0.022, 0.206]$&AC2&$[ 1.373, 5.093]$\\
                                              &$-$&$-$&$-$&ACL&$[ 1.368, 5.080]$\\

\hline
\end{tabular}\caption{Comparisons of CIs  in the analysis of a bioassay assessing the carcinogenic effect on four sex-strain groups of mice\newline
$^{\S}$ MH estimate of RD and RR \newline
$^{\dagger}$ Based on dually consistent variance estimate under both large strata and sparse data
}\label{tbin}}
\end{table}

\begin{example}\label{miceexam}
\normalfont
We revisit the bioassay data displayed  in Section \ref{dmover}. Appendix \ref{biaseff} shows that the estimated difference is generally biased in the unstratified analysis.
Stratified analysis can prevent bias, and  improve the precision of the treatment effect estimate \cite{grizzle:1982,mchugh:1983, miratrix:2013, tang:2018, tang:2020b}, and are hence  recommended.
The MOVER CIs are compared with the asymptotic CIs on the basis of  the  variance
valid only under  large strata  (labeled as ``ASY'' or ``Wald'', formulae given in the appendix) and the  dually consistent variance estimates (labeled as ``DC'') 
valid for both sparse data and  large strata \citep{greenland:1985, robins:1986, sato:1989}. The asymptotic CIs for
 RR  are obtained by exponentiating the CIs on the log scale. Sample SAS code is provided in the Supplementary Material.

The results are displayed in Table \ref{tbin}, and  all methods adjust for the stratification factor.
The following CIs require the constant effect assumption: 1)  All DC CIs for the MH estimators;
 2) All CIs for the INV weighted RD. 
These assumptions can not hold simultaneously.
For illustration purposes,  we assume the relevant assumptions hold for each method.
Throughout this paper, all CIs $(\mathcal{L},\mathcal{U})$ for the MR weighted RD are corrected as $(\mathcal{L}-c,\mathcal{U}+c)$ to penalize for ignoring variabilities in the weight  \citep{mehrotra:2000} even if the RD is constant, 
 where $c=\frac{3}{16} [\sum_{s=1}^S \frac{n_{s1}n_{s0}}{n_{s1}+n_{s0}}]^{-1}$.
No continuity correction is applied to other CIs.

On both RD and RR effect metrics, the AC and AC2 approaches yield very similar CIs, and the AV CI has a larger lower confidence limit than other CIs.
On the RR metric, the AVL CIs are wider than other MOVER CIs.

The Wald CI for the INV-weighted RD contains 0. All other CIs indicate there is a significant difference between two treatment groups as the lower limits for RD are above 0, and the lower limits for RR are above 1.
\end{example}

\begin{figure}[htbp]
\centering
\subfigure[male: $(n_{11},n_{10})=(90,81)$]{
\includegraphics[scale=0.4]{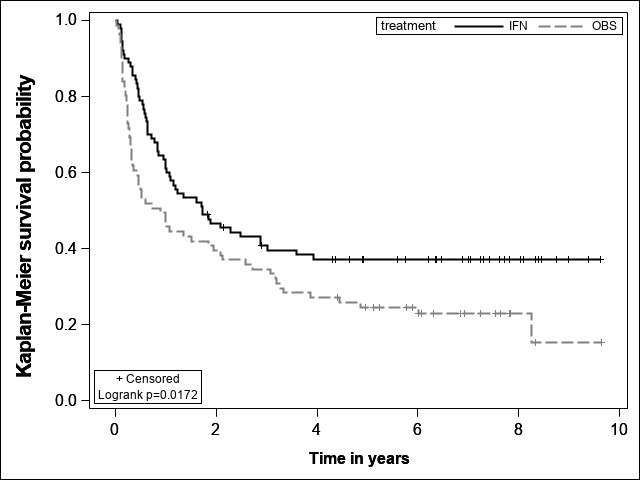}
}
\subfigure[female: $(n_{21},n_{20})=(54,59)$]{
\includegraphics[scale=0.4]{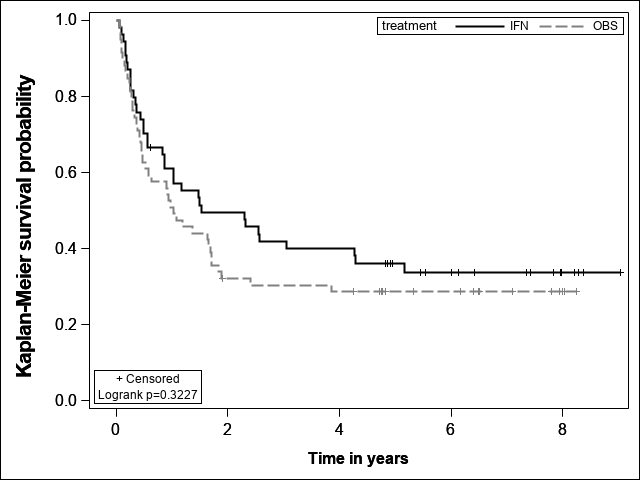}
}
\caption{KM plot of survival curves by sex in the ECOG 1684 trial}\label{kmplot}
\end{figure}

\begin{table}[h]
\centering{
\begin{tabular}{c@{\extracolsep{3pt}}c@{\extracolsep{3pt}}c@{\extracolsep{3pt}}c@{\extracolsep{3pt}}c@{\extracolsep{3pt}}c@{\extracolsep{3pt}}c@{\extracolsep{3pt}}c@{\extracolsep{3pt}}c@{\extracolsep{3pt}}c@{\extracolsep{3pt}}c@{\extracolsep{3pt}}c@{\extracolsep{3pt}}}\\\hline
 \multicolumn{2}{c}{One-sample CI} &     \multicolumn{2}{c}{Difference } &     \multicolumn{4}{c}{Ratio} \\  \cline{1-2} \cline{3-4} \cline{5-8}
  obs: control & IFN: treated &   AV &   AC$^{\dagger}$  & AV & AVL  &   AC$^{\dagger}$ & ACL \\ \hline
   \multicolumn{8}{l}{\bf 8-year KM survival probability: $(\hat\delta_{10},\hat\delta_{11},\hat\delta_{20},\hat\delta_{21})=( 0.228, 0.372, 0.286, 0.338)$}\\
                                        $[ 0.148, 0.332]$&$[ 0.278, 0.476]$&$[-.001, 0.211]$&$[-.001, 0.212]$&$[ 1.004, 2.023]$&$[ 0.982, 1.994]$&$[ 0.997, 2.058]$&$[ 0.996, 2.055]$\\
                                          $[ 0.186, 0.412]$&$[ 0.223, 0.473]$&$[ 0.000, 0.212]$&$[-.000, 0.214]$&$[ 1.007, 2.033]$&$[ 0.985, 2.003]$&$[ 1.000, 2.067]$&$[ 1.000, 2.065]$\\
                                          \\\multicolumn{8}{l}{\bf 8-year RMST: $(\hat\delta_{10},\hat\delta_{11},\hat\delta_{20},\hat\delta_{21})=( 2.692, 3.644, 2.874, 3.527)$}\\
                                       $[ 2.073, 3.439]$&$[ 2.969, 4.378]$&$[ 0.049, 1.584]$&$[ 0.048, 1.595]$&$[ 1.020, 1.654]$&$[ 1.011, 1.646]$&$[ 1.015, 1.667]$&$[ 1.015, 1.666]$\\
                                          $[ 2.124, 3.795]$&$[ 2.683, 4.471]$&$[ 0.053, 1.587]$&$[ 0.051, 1.598]$  &$[ 1.021, 1.657]$&$[ 1.012, 1.648]$&$[ 1.016, 1.669]$&$[ 1.016, 1.668]$\\
                        
\hline
\end{tabular}\caption{MOVER CIs in the analysis of  the ECOG 1684 trial stratified by sex. \newline
[1] For each effect measure, the first (second) row displays the one-sample CI for  male (female) and the CIs for the difference and ratio on basis of the MH (INV) weight. \newline
$^{\dagger}$ The AC2 CI is very similar to the AC CI for both difference and ratio, and not displayed due to limited space 
}\label{trmstkm}
}
\end{table}

\begin{example}\label{survexam}
\normalfont
We apply the proposed method to the RMST and milestone survival for time to event outcomes.
We analyze the  Eastern Cooperative Oncology Group (ECOG) 1684 trial \citep{kirkwood:1996}. 
Patients with American Joint Committee on Cancer stage IIB or III melanoma were randomized to receive Interferon alfa-2b (IFN) or to receive close observation (Obs).
A total of 287 patients were accrued between 1984 and 1990. and remained blinded under analysis until 1993. 
Patients in the IFN arm had significantly improved relapse-free
survival (RFS)  compared with Obs. 

The purpose of the analysis is to compare the 8-year KM survival  and RMST for RFS between two groups stratified by sex.
Figure \ref{kmplot} plots the KM curves for RFS by sex. 
Slightly earlier and larger separation in the KM curves was observed in male patients than in female subjects. 
 We first illustrate how to assess the treatment by stratum interaction by the MOVER technique.
The one-sample CI for both KM survival  and RMST is obtained by inverting the so-called  score-type or constrained variance test
in the sense that the variance of $\hat{\delta}_{gs}$ is obtained under the null hypothesis \citep{thomas:1975, barber:1999, tang:2021}.
Table \ref{trmstkm} presents the single-sample score type CIs for the KM survival \citep{barber:1999} and RMST   \citep{tang:2021} by sex.
The point estimate ($95\%$ MOVER CI based on the score limit) for the difference of RMST is $\hat{\delta}_{11}-\hat\delta_{10}=0.953$ ($[-0.054,1.912]$) for male, 
and $\hat{\delta}_{21}-\hat\delta_{20}=0.653$ ($[-0.596, 1.858]$) for female. Another application of the MOVER technique
yields the $95\%$ CI $[-1.271,1.874]$  for $({\delta}_{11}-\delta_{10})- ({\delta}_{21}-\delta_{20})$, which contains 0.
Although a larger separation in the KM curves is observed among males than among females,  the  difference does not significantly  differ between male and female. 
In general, a test of the treatment by stratum interaction requires a much larger sample size.
A similar technique may be employed on the log scale to assess whether  the ratio effect measure differs between male and female.

Table \ref{trmstkm}  displays the stratified MOVER CIs for the difference and ratio of the milestone survival and RMST. 
The MH weight is $(w_1,w_2)=(0.602,0.398)$. The INV weight is $(w_1,w_2)=(0.614,0.386)$ for the milestone survival, 
and $(w_1,w_2)=(0.613,0.387)$ for RMST. The INV and MH weights are close to each other, and yield very similar CIs. 
For RMST, the lower limit  is above 0 for the difference, and above 1 for the ratio, evidencing the superiority of IFN over control in prolonging the mean survival.
The difference for the milestone survival is marginally significant in the sense that the lower limit  is near 0 for the difference, and near 1 for the ratio. 
\end{example}

\section{Simulation}\label{secsim}
We conduct simulation to assess the proposed MOVER CIs for RD and RR in stratified analyses of binary proportions, and compare them with the asymptotic CIs.
In all cases, one million datasets are  simulated. 
A dataset will be regenerated until the following conditions hold:  1) there is at least one event in the study  ($\sum_{s=1}^S [\hat{p}_{s0}+\hat{p}_{s1}]>0$) on the RD metric,
2) there is at least one event in each treatment group ($\sum_{s=1}^S \hat{p}_{s0}>0$, $\sum_{s=1}^S \hat{p}_{s1}>0$) on the RR metric. 
Otherwise the DC CIs for the MH estimators \citep{tang:2020b} and some MOVER CIs may be incomputable.

There is more than $95\%$ chance that the empirical coverage probability (CP) lies within $0.04\%$ of the true value when the target level is $0.95$.
This greatly reduces the random error in the estimate, and enables the detection of subtle differences between methods. 
When the sample size is not large, it is usually difficult to control the CP exactly at the target level.
Following Tang \cite{tang:2020b}, we deem the CPs to be highly satisfactory if it lies within $[0.945, 0.955]$ at the nominal $95\%$ level.

On the RD scale, we use the estimate $\hat{p}_{sg}=0.5/n_{sg}$  if $x_{sg}=0$, and $\hat{p}_{sg}=1-0.5/n_{sg}$ if $x_{sg}=n_{sg}$ 
to calculate  the INV and MR weights \citep{greenland:1985}, but the difference is still computed as  $\hat{p}_{s1}-\hat{p}_{s0}=x_{s1}/n_{s1}- x_{s0}/n_{s0}$.

A CI is said to be asymptotically valid if the CP converges to $1-\alpha$ as $n\rightarrow \infty$. Simulation may show that a method is not asymptotically valid. 
In Example \ref{examheter} below,  the CPs 
are not close to $1-\alpha=95\%$ in the DC method for RD and RR under heterogeneity effect, and  in the YS method for RD even when $n=8000$.
In Example \ref{examyan} below, the type I errors are not close to $\alpha=5\%$  in the MR-weighted tests for RD at $n= 10000$.
These methods may not be asymptotically accurate.   All other MOVER methods are asymptotically valid according to the theory in Section \ref{dmover}.

\begin{example}\label{cpsim}\normalfont
We conduct a  simulation study to assess CIs for stratified comparisons of binary outcomes for a constant effect on the RD and RR scales. 
Suppose there are two strata.  The sample sizes are either balanced  $(n_{10},n_{11},n_{20},n_{21})=(24,24,16,16)$ or unbalanced $(n_{10},n_{11},n_{20},n_{21})=(12,36,8,24)$ in the two groups. The trial size is $n=80$.
The true response rates  in the control arm are set to  $(p_{10},p_{20})= (0.12k_1,0.12k_2)$ for $k_1,k_2=1,2,3,4,5$.
There are a total of $25$  combinations of $(p_{10},p_{20})$. The true effect is 0 or 0.3 on the RD scale, and 1 or 1.5 on the RR metric.
The results are displayed in Figures \ref{binplot}a, \ref{binplot}c and \ref{binplot}e. 

{\bf RD metric:}  The result for the INV weight  is not reported due to limited space.
The DC and Wald CIs  tend to undercover, and the coverage becomes much worse under  unbalanced sample sizes.
While both AC and AC2 maintain  CP above $94.5\%$ in most cases, the AC generally yields  slightly larger CP than the AC2.
Although the AV method performs better  than the DC and Wald methods, it yields lower than nominal CPs 
in quite many cases when the sample sizes are  unequal.

{\bf RR metric:} For the MH estimator, the DC and ASY methods yield CPs below the nominal level in about half of cases at $RR=1.5$ under unbalanced sample sizes, and work well in other situations.
The AC and ACL methods produce similar results that are above $94.5\%$ in nearly all cases.
The AC2 CI is slightly less conservative than the AC CI.
The AV  method gives CPs that are slightly below $94.5\%$ in a few cases with unbalanced sample sizes. 
The AVL method yields CPs above $97.5\%$  more often than other three MOVER CIs.

\end{example}

\begin{figure}[htbp]
\centering
\subfigure[Exam 3: 2 strata, MH weight, $n=80$]{
\includegraphics[scale=0.64]{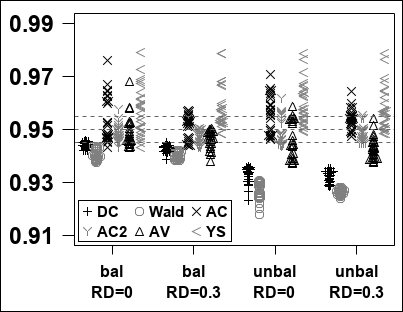}
}
\subfigure[Exam 4:  3 strata, MH weight, $n=96$]{
\includegraphics[scale=0.64]{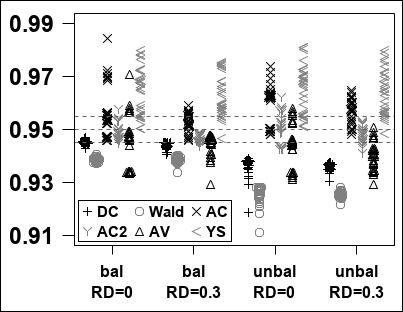}
}
\subfigure[Exam 3: 2 strata, MR weight, $n=80$]{
\includegraphics[scale=0.64]{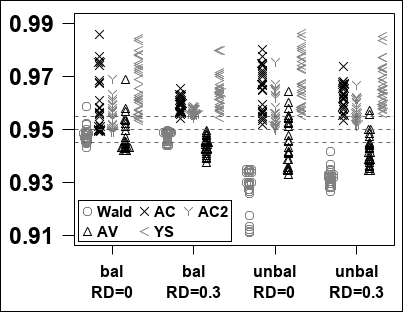}
}
\subfigure[Exam 4:  3 strata, MR weight, $n=96$]{
\includegraphics[scale=0.64]{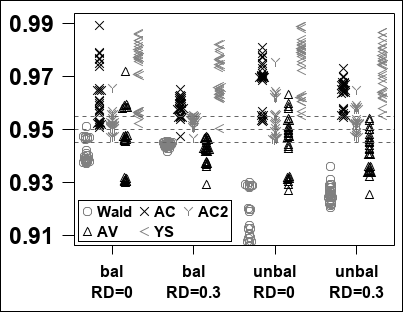}
}
\subfigure[Exam 3: 2 strata, MH estimate, $n=80$]{
\includegraphics[scale=0.64]{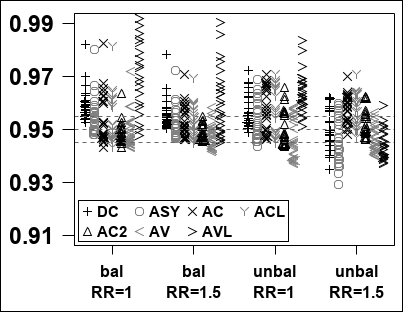}
}
\subfigure[Exam 4: 3 strata, MH estimate, $n=96$]{
\includegraphics[scale=0.64]{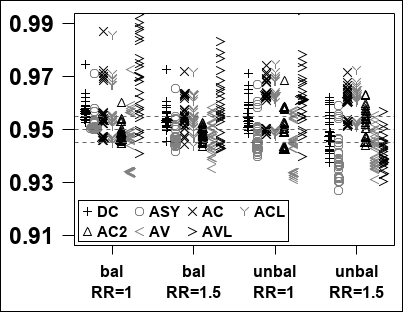}
}
\caption{Coverage probabilities of various CIs for RD and RR when the effect is homogeneous across strata}\label{binplot}
\end{figure}

\begin{example}\label{cpsim2}\normalfont
We conduct a  simulation study to assess CIs for stratified comparisons of binary outcomes for a constant effect on the RD and RR scales
when there are three strata.  The sample sizes are either balanced  $(n_{10},n_{11},n_{20},n_{21},n_{30},n_{31})=(20,20,16,16, 12,12)$ or
 unbalanced $(n_{10},n_{11},n_{20},n_{21},n_{30},n_{31})=(10,30,8,24,6,18)$ in the two groups. The trial size is $n=96$.
The true response rates  in the control arm are set to  $(p_{10},p_{20},p_{30})= (0.12k_1,0.12k_2, 0.12k_3)$ for $k_1,k_2,k_3=1,3,5$.
There are a total of $27$  combinations of $(p_{10},p_{20},p_{30})$.
The results are displayed in Figures \ref{binplot}b, \ref{binplot}d and \ref{binplot}f, and the result pattern is fairly similar to that observed in Example \ref{cpsim}.
\end{example}

\begin{figure}[htbp]
\centering
\subfigure[ $RD_1=0$, $RD_2=0.3$, MH weight]{
\includegraphics[scale=0.64]{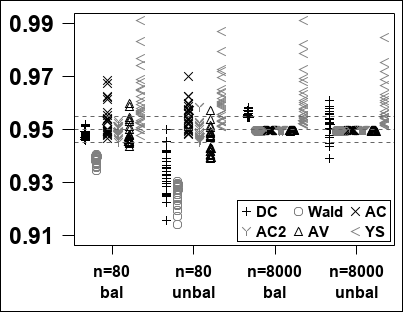}
}
\subfigure[$RD_1=0$, $RD_2=0.3$, MR weight]{
\includegraphics[scale=0.64]{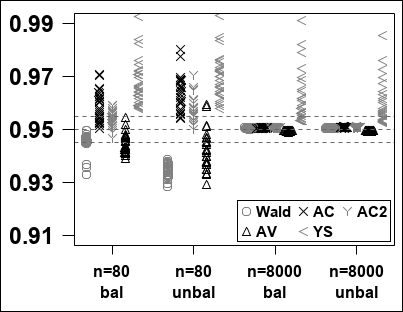}
}
\subfigure[ $RR_1=1$, $RR_2=1.5$, MH weight]{
\includegraphics[scale=0.64]{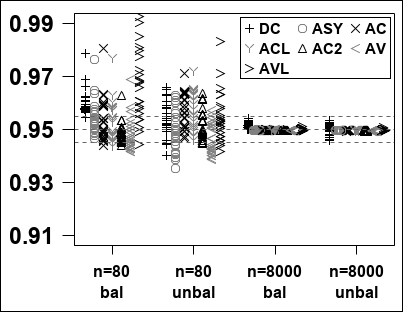}
}
\caption{Coverage probabilities of various CIs for RD and RR under heterogeneity of effects across strata}\label{binplotnon}
\end{figure}

\begin{example}\label{examheter}\normalfont
Simulation is conducted to demonstrate that the MOVER CIs for the MH estimator of RD and RR do not require the assumption of a constant effect. 
For RR, the MH estimate  is the ratio of 
the MH weighted proportions between two groups. When the group size ratio $n_{s1}/n_{s0}$ is constant across strata (commonly used  in large clinical trials),
 the MH estimate is identical to the crude RR estimate from the unstratified analysis, but a stratified analysis improves the
precision of the RR estimate \citep{greenland:1985}.

We use a similar simulation setup to  example \ref{cpsim} except for the following differences.
The true RD is $(0,0.3)$, or the true RR is $(1,1.5)$  in the two strata.
The total  size is either $n=80$ or  $8000$, but the group size ratio remains the same. 
We assess whether each CI contains the pooled difference $w_1^{(mh)}\mbox{RD}_1 +(1-w_1^{(mh)}) \mbox{RD}_2=0.12$ on the RD metric,
or the ratio of the pooled proportion $\frac{w_1^{(mh)}p_{11}+(1-w_1^{(mh)})p_{21}}{w_1^{(mh)}p_{10}+(1-w_1^{(mh)})p_{20}}$ on the RR metric,
where $w_1^{(mh)}=0.6$. Figures  \ref{binplotnon} shows the results. The result for the MR weight on the RD metric is also reported since 
the MR weighting scheme is specially designed to handle non-constant RDs \citep{mehrotra:2000}. 

At $n=8000$, all CIs  except the YS method for RD and the DC methods for both RD \citep{sato:1989} and RR \citep{greenland:1985}
maintain CP in a very narrow range around $95\%$, indicating the asymptotic validity of these methods in the presence of a non-constant RD or RR.
The YS method is generally conservative. The DC methods require the constant effect assumption.
When the RDs differ between the two strata, the MR weight defined in Equation \eqref{mrweightdef} in Example \ref{examyan} below approximates $f_1=w_1^{(mh)}$ when the sample size 
is extremely large  (constant allocation ratio across strata),
and the CPs are hence well maintained around the target level.

At n=80, the ASY or Wald methods tend to undercover. The AC and AC2 methods show better performance than other approaches on the RD scale,
and the AC, AC2 and AVL outperform other methods under unequal randomization on the RR metric. 
The AC2 CI appears to be slightly less conservative than the AC CI on both metrics.

\end{example}

\begin{example}\label{examyan}
\normalfont
We compare the type I error and power for testing the hypothesis $H_0: \text{RD}=0$  on the basis of the Wald  and MOVER CIs, and
 illustrate that the MOVER method may not be suitable for the MR weight in estimating the RD.
 The true proportion  is $(p_{10},p_{20})=(0.1,0.6)$, and RD=0 or 0.05 is constant across strata.

The result is presented in Table \ref{tableyan}. 
The YS \cite{yan:2010} approach yields the two-sided type I error far below the target  $5\%$ level under $H_0: RD=0$, and produces  lower power than other approaches at $RD=0.05$ in all cases.
 For both MH and INV weights, the AC, AC2 and AV approaches have nearly nominal  type I error rates,
evidencing  the asymptotic validity of these MOVER approaches for both fixed and random weights.

The MR weight  slightly inflates the type I error rates  at $n_{11}=n_{10}=n_{21}=n_{20}=10000$ and $RD=0$. 
In this scenario, the empirical  mean  weight $\pm$ standard deviation in stratum $1$ is $0.6490\pm 0.0602$  for the MR weight, 
and   $0.7273\pm  0.0038$ for the INV weight based on $10^6$ simulated datasets. For a constant RD, the MR weight \citep{mehrotra:2000}
\begin{equation}\label{mrweightdef}
 \hat{w}_{mr_1}=\frac{\hat{V}_2+f_1(\hat\delta_2-\hat\delta_1)^2}{\hat{V}_1+\hat{V}_2+(\hat\delta_2-\hat\delta_1)^2}
\end{equation}
 does not actually converge to the limit of the INV weight  $w_{inv_1}=\frac{V_1^{-1}}{V_1^{-1}+V_2^{-1}}=\frac{V_2}{V_1+V_2}$ 
since $\frac{(\hat\delta_2-\hat\delta_1)^2}{\hat{V}_1+\hat{V}_2}$ does not converge in probability to 0, where $\delta_s=p_{s1}-p_{s0}$, 
 $V_s=p_{s1}(1-p_{s1})/n_{s1}+ p_{s0}(1-p_{s0})/n_{s0}$ and
$f_1= (n_{11}+n_{01})/\sum_s (n_{s1}+n_{s0})$. 
The MR weight is not as powerful as the INV weight in detecting a constant RD. 

Lemma \ref{varsigle} is suitable for the INV weight, but not  for the MR weight. The MR weight depends on both  $f_{sg}=n_{sg}/n$ and $n$, but the INV weight depends only on $f_{sg}$'s.
At $n_{11}=n_{10}=n_{21}=n_{20}=10000$, the type I error rate is inflated in the MR-weighted approaches
because the variability in the MR weight is ignored, and the continuity correction of Mehrotra and Railkar \cite{mehrotra:2000} has a negligible effect on the inference in extremely large samples.
\end{example}

\begin{table}[h]\caption{Empirical type I error and power based on $10^6$ simulated datasets for testing  $H_0: \mbox{RD}=0$ at a significance level of $\alpha=0.05$ in a study with $2$ strata  \newline
[1] The sample sizes are the same in each group across strata \newline
[2] The correction method of Mehrotra and Railkar (2000) is applied to all CIs for the MR weight}\label{tableyan}
\centering{
\begin{tabular}{c@{\extracolsep{3pt}}c@{\extracolsep{3pt}}c@{\extracolsep{3pt}}c@{\extracolsep{5pt}}c@{\extracolsep{3pt}}c@{\extracolsep{3pt}}c@{\extracolsep{5pt}}c@{\extracolsep{3pt}}c@{\extracolsep{3pt}}c@{\extracolsep{5pt}}c@{\extracolsep{3pt}}c@{\extracolsep{3pt}}c@{\extracolsep{3pt}}c@{\extracolsep{3pt}}c@{\extracolsep{3pt}}c@{\extracolsep{3pt}}c@{\extracolsep{3pt}}c}\\\hline

  &\multicolumn{5}{c}{MH weight} &     \multicolumn{5}{c}{INV weight } &\multicolumn{5}{c}{MR weight }  \\ \cline{2-6}\cline{7-11} \cline{12-16}
$n_{sg}$ & Wald & YS & AV & AC & AC2  & Wald & YS & AV & AC & AC2 & Wald & YS & AV & AC & AC2  \\ \hline
\multicolumn{16}{c}{two-sided type I error ($\%$): true RD=0} \\
                                                                             50& 5.43& 2.01& 4.79& 4.64& 4.88& 5.53& 1.19& 4.44& 4.96& 5.08& 5.36& 1.42& 5.21& 4.61& 4.83\\
                                                                              100& 5.24& 2.21& 4.91& 4.88& 4.97& 5.25& 1.20& 4.78& 5.16& 5.10& 5.39& 1.49& 5.50& 5.12& 5.13\\
                                                                              200& 5.13& 2.13& 4.98& 4.97& 5.00& 5.15& 1.18& 4.89& 5.09& 5.05& 5.44& 1.60& 5.62& 5.31& 5.32\\
                                                                              500& 5.10& 2.19& 5.04& 5.04& 5.05& 5.20& 1.21& 5.10& 5.17& 5.16& 5.56& 1.63& 5.74& 5.52& 5.52\\
                                                                            10000& 5.03& 2.15& 5.03& 5.03& 5.03& 5.05& 1.18& 5.05& 5.05& 5.05& 5.67& 1.71& 5.74& 5.67& 5.67\\
\multicolumn{16}{c}{power  ($\%$): true RD=0.05} \\
                                                                               50& 14.2&  7.5& 13.4& 13.1& 13.6& 15.8&  6.5& 14.4& 15.1& 15.2& 14.4&  6.4& 14.7& 13.4& 13.7\\
                                                                              100& 23.0& 13.9& 22.5& 22.3& 22.6& 25.9& 12.7& 25.0& 25.5& 25.5& 24.3& 12.6& 24.9& 23.8& 23.9\\
                                                                              200& 40.2& 28.3& 40.0& 39.8& 40.0& 45.4& 27.3& 44.8& 45.2& 45.1& 43.2& 27.1& 44.0& 42.9& 43.0\\
                                                                              500& 77.1& 65.8& 77.0& 77.0& 77.0& 82.7& 67.8& 82.6& 82.7& 82.7& 80.7& 66.7& 81.1& 80.6& 80.6\\
\hline
\end{tabular}}
\end{table}

\section{Discussion}
We propose several MOVER CIs for a difference or ratio parameter under stratified sampling.
These approaches require only the single sample point estimate, variance estimate and CI for the parameter of interest in each stratum, and hence can be easily applied to different outcomes.
For the difference parameter,  either the CI for the weighted difference, or the CI for the weighted effect in each group, or the variance for the weighted difference is 
calculated as the sum of the corresponding stratum-specific statistics. The CIs for the ratio are derived by the Fieller or log-ratio approaches. 
The Fieller CI for the ratio of proportion, rate or mean survival can be uniquely determined and non-disjoint. All these interval estimation methods except the AC2 CI for a ratio are non-iterative.

We apply the MOVER CIs to the binary and survival outcomes. As demonstrated by several simulation studies, the proposed MOVER approaches generally outperform the asymptotic CIs.  
The YS\cite{yan:2010} MOVER CI for RD is  conservative particularly when the proportions vary greatly across strata.
In general, the AC and AC2 approaches show better performance than the AV approach for both difference and ratio parameters. 
For binary outcomes, the AC2 CI tends to be slightly less conservative for the MH estimators of RD and RR than the AC CI.
For the ratio, the Fieller CI is preferable over the log-ratio CI. The log-ratio CI is incomputable if the point estimate $\hat\tau_1$ or $\hat\tau_0$ is 0.
Otherwise, the ACL and AC CIs are generally similar. The AVL approach appears to be more conservative than the AV method. 
In summary, we recommend the AC and AC2 interval estimation methods for the difference, and the Fieller approach for the ratio.

The proposed MOVER CIs are asymptotically valid for random weights  under the assumption of a constant difference or ratio across strata, or for non-constant effects when the weights are fixed.
However, this does not work for the MR-weighted RD for binary outcomes, and the reasons are given in Example \ref{examyan}.
 The continuity correction suggested by Mehrotra and Railkar \cite{mehrotra:2000} for penalizing for ignoring variability in the MR weight
improves the performance in small and moderate samples, but the CP of the CI may slightly deviate from the nominal level when the RD is constant in extremely large samples.

\appendix

\section{Ignoring stratification may lead to biased estimate}\label{biaseff}
We assess the bias of the unstratified RD estimate in data with two strata.
The true proportion is $p_{sg}$ in group $g=0,1$ stratum $s=1,2$. The  true RD between two groups is $\Delta$ in both strata. That is,
$p_{11}-p_{10}= p_{21} - p_{20} =\Delta$. 

Let $n_{sg}$ be the number of subjects in group $g=0,1$ stratum $s=1,2$. Let $r_s= n_{s1}/n_{s0}$.  In the unstratified analysis, the expectation of the estimated RD
is
 $$\frac{ n_{11} p_{11}+n_{21} p_{21} }{ n_{11}+n_{21}} - \frac{ n_{10} p_{10}+n_{20} p_{20} }{ n_{10}+n_{20}} =\Delta +  \frac{ n_{10}r_1 p_{10}+n_{20} r_2p_{20} }{ n_{10}r_1+n_{20}r_2}- \frac{ n_{10} p_{10}+n_{20} p_{20} }{ n_{10}+n_{20}} =\Delta + \frac{n_{10}n_{20}(r_1-r_2)(p_{10}-p_{20})}{ (n_{10}r_1+n_{20}r_2)(n_{10}+n_{20})}.$$
The expected difference in the unstratified analysis is $\Delta$ only if $p_{10} = p_{20}$ (i.e. no stratification effect) or $r_1 = r_2$ (the sample size ratio between two groups is constant in the two strata).
If the risk in the control arm differs  between two strata (i.e. $p_{10}\neq p_{20}$), but  there is imbalance in the stratification factor between two groups (equivalently $r_1\neq r_2$),
 the unstratified analysis is biased.

\section{Proof of two lemmas}
\begin{proof}[Proof of Lemma \ref{varsigle}]
By the delta method, we get Lemma \ref{varsigle}a since
$$\mbox{var}(\hat\tau_g)=\sum_s A_s \mbox{var}(\hat\thetav_s) A_s'=\sum_{s} w_s^2 B_s\mbox{var}(\hat\thetav_s)B_s'=\sum_{s} w_s^2 \mbox{var}(\hat\delta_{sg}),$$ 
where $B_s=\frac{\partial{\delta}_{sg}}{\partial \thetav_s}$ and 
$A_s= [\sum_{s=1}^S \frac{\partial{w_s}}{\partial \thetav_s} (\hat{\delta}_s- \delta_s)+ \hat{w}_s \frac{\partial{\delta_s}}{\partial \thetav_s}]|_{\hat{\thetav}_s=\thetav_s}=  w_s B_s$.
Lemma \ref{varsigle}b holds  obviously for fixed weights, and can be proved similarly by writing $\hat\tau=\hat{w}_s[(\hat\delta_{s1}-\delta_{s1})-\phi(\hat\delta_{s0}-\delta_{s0})]$ for random weights. 
Lemma \ref{varsigle}b implies that $\mbox{cov}(\hat\tau_1^*,\hat\tau_0^*)=0$. 
Note that $c=\frac{\sum_s \hat{w}_s \delta_{s1}}{\tau_1}-\frac{\sum_s \hat{w}_s \delta_{s0}}{\tau_0}=0$ when $\hat{w}_s$ is constant or when $\delta_{s1}=\phi\delta_{s0}$.
Another application of the delta method yields 
$ \mbox{var}(\hat\psi)=\text{var}(\frac{\hat\tau_1}{\tau_1}- \frac{\hat\tau_0}{\tau_0} -c)=\text{var}(\frac{\hat\tau_1^*}{\tau_1}- \frac{\hat\tau_0^*}{\tau_0})= \frac{\text{var}(\hat\tau_1^*)}{\tau_1^{2}} + \frac{ \text{var}(\hat\tau_0^*)}{\tau_0^2}$.

\end{proof}

\begin{proof}[Proof of Lemma \ref{eachci}] It is easy to see the following asymptotic relationships
$$(L_g-\hat\tau_g)^2=[\sum_s\hat{ w}_s (l_{sg}^{(\gamma)}- \hat{\delta}_{sg})]^2= (z_{\gamma/2} \sum_s\hat{ w}_s  \hat\sigma_{sg})^2 =z_{\alpha/2}^2\sum_s w_s^2 \sigma_{sg}^2+o_p(n^{-1}) $$
$$(\log(L_g)-\log(\hat\tau_g))^2= (\frac{L_g-\hat\tau_g}{\hat\tau_g})^2+o_p(n^{-1})= \frac{z_{\alpha/2}^2 \sum_s w_s^2 \sigma_{sg}^2}{\tau_g^{2}}+o_p(n^{-1}).$$
The last equation is obtained by Taylor series expansion of $\log(1+x)$ around $x=\frac{L_g-\hat\tau_g}{\hat\tau_g}$. The remaining relationships can be proved similarly.
\end{proof}

\section{Large sample variance for MH estimates}
Let $w_s=n_{s1}n_{s0}/n_{s}$. The delta variance for the weighted RD $\hat\tau=\sum_s w_s (\hat{p}_{s1}-\hat{p}_{s0})$ is given by Lemma \ref{varsigle}b
$$\text{var}(\hat\tau)= \sum_s w_s^2 [{p}_{s1}(1-{p}_{s1})/n_{s1}+{p}_{s0}(1-{p}_{s0})/n_{s0}].$$
For a constant  RD$=\delta$, replacing ${p}_{s1}(1-{p}_{s1})/n_{s1}$ by $[({p}_{s0}+\delta)(1-{p}_{s1})+ {p}_{s1}(1-{p}_{s0}-\delta)]/(2n_{s1})$
and ${p}_{s0}(1-{p}_{s0})/n_{s0}$ by $[({p}_{s1}-\delta)(1-{p}_{s0})+ ({p}_{s1}-\delta)(1-{p}_{s0})]/(2n_{s0})$ yields the dually consistent variance \citep{sato:1989} with the MH weight.

By Lemma \ref{varsigle}c, the asymptotic variance for the MH estimator of $\log(\widehat{RR}) =\log[\sum_s w_s\hat{p}_{s1}]-\log[\sum_s w_s\hat{p}_{s0}]$ is given by 
$$\text{var}(\log(\widehat{RR})) =\frac{\sum_s w_s^2 {p}_{s1}(1-{p}_{s1})/n_{s1}}{(\sum_s w_s {p}_{s1})^2} + \frac{\sum_s w_s^2 {p}_{s0}(1-{p}_{s0})/n_{s0}}{(\sum_s w_s {p}_{s0})^2}.$$
It  can be reorganized as the dually consistent variance  if RR is constant across strata \citep{greenland:1985}
$$\text{var}(\log(\widehat{RR})) =\frac{\sum_s w_s^2 {p}_{s0}(1-{p}_{s1})/n_{s1}}{(\sum_s w_s{p}_{s1})(\sum_s w_s {p}_{s0})}
 + \frac{\sum_s w_s^2 {p}_{s1}(1-{p}_{s0})/n_{s0}}{(\sum_s w_s {p}_{s0})(\sum_s w_s{p}_{s1})} =  \frac{\sum_s w_s ( \bar{p}_{s}-p_{s0}p_{s1}) }{(\sum_s w_s {p}_{s0})(\sum_s w_s {p}_{s1})} $$
where $\bar{p}_s= \text{E}(\frac{x_{s1}+x_{s0}}{n_s})=\frac{n_{s0}p_{s0}+n_{s1}p_{s1}}{n_s}$. The dual consistent variance estimator is computed by replacing
$p_{sg}$ by  $\hat{p}_{sg}=\frac{x_{sg}}{n_{sg}}$.
\bibliography{moverci} 

\end{document}